\documentclass[]{ceurart}

\usepackage{subcaption}
\usepackage{xspace}
\usepackage{amsthm}
\usepackage{algorithm}
\usepackage[noend]{algpseudocode}
\usepackage{placeins}

\graphicspath{{figures/}}

\newtheorem{definition}{Definition}
\newtheorem{lemma}{Lemma}
\newtheorem{proposition}{Proposition}
\newtheorem{theorem}{Theorem}
\newtheorem{corollary}{Corollary}
\newtheorem{example}{Example}

\newcommand{\sharpalg}{\textup{\textsc{SHARP}}\xspace}
\newcommand{\asharpalg}{\texorpdfstring{\textup{A$\sharp$}}{A-sharp}\xspace}
\newcommand{\Tuple}[1]{\left\langle #1 \right\rangle}

\makeatletter
\providecommand{\@LN}[2]{}
\providecommand{\@LN@col}[1]{}
\makeatother

\begin{document}

\copyrightyear{2026}
\copyrightclause{Copyright for this paper by its authors.
  Use permitted under Creative Commons License Attribution 4.0
  International (CC BY 4.0).}

\conference{Joint Workshop on Planning for Complex Real-World Applications (CAIPI) and Bridging the Gap Between AI Planning and (Reinforcement) Learning (PRL), co-located with IJCAI-ECAI 2026, August 17, 2026, Bremen, Germany}

\title{Dynamic Haven Selection for Multi-Agent Pickup and Delivery in Constrained Warehouses}

\author[1]{Taisei Hirayama}[%
orcid=0009-0001-6886-293X,
email=hirayama.h77@gmail.com,
]
\cormark[1]
\author[2]{Kohei Yoshida}[%
orcid=0009-0006-8844-4889,
email=kohei.yoshida.ab@mail.toyota-shokki.co.jp,
]
\author[1]{Hiroki Sakaji}[%
orcid=0000-0001-5030-625X,
email=sakaji@ist.hokudai.ac.jp,
]
\author[1]{Itsuki Noda}[%
orcid=0000-0003-1987-5336,
email=i.noda@ist.hokudai.ac.jp,
]

\address[1]{Hokkaido University, Sapporo 060-0808, Japan}
\address[2]{Toyota Industries Corporation, Aichi 474-8601, Japan}
\cortext[1]{Corresponding author.}

\begin{abstract}
Space-efficient warehouse layouts often contain single-agent-width aisles and dead-end workstations where robots have few places to wait without blocking others.
In Multi-Agent Pickup and Delivery (MAPD) on such constrained layouts, robots must accept online pickup-delivery tasks while preserving protected waiting locations, called Havens, to avoid collisions and deadlocks.
The Safe HAven Retreat Planner (\sharpalg) introduced a mechanism that extends each committed task path with a validated retreat to the agent's dedicated initial Haven.
These fixed-Haven commitments can still send agents toward distant Havens after deliveries.
We present \asharpalg (Adaptive \sharpalg, pronounced ``A-sharp''), a dynamic Haven selection method that allows an agent's retreat target to change at task assignment time.
The central difficulty is that a naive switch can make two agents rely on the same waiting location or let another committed path pass through a location that is still occupied or reserved.
\asharpalg addresses this difficulty through an availability test for candidate Havens and a pending-release rule that keeps the previous Haven protected until the agent actually departs.
Under explicit Haven-structure and Safe Interval Path Planning (SIPP) assumptions, we prove invariant preservation and finite-release completeness, meaning that every task in any finite release sequence is delivered.
Across 72,000 runs on 14,400 paired map--agent-count--rate--seed cases over four maps, both \sharpalg and \asharpalg complete all 14,400 runs assigned to each method.
For makespan (final delivery time), a prespecified paired comparison with Holm correction over all 138 Haven-surplus configurations ($|A|<|H|$) finds that \asharpalg is significantly better in 107 configurations and never significantly worse than \sharpalg; on the tested tree map, the median reduction is 16.7\%.
Reductions in average release-to-delivery time are strong on the tested tree map but not uniform on the two narrow-biconnected maps.
\end{abstract}

\begin{keywords}
Multi-Agent Pickup and Delivery \sep
Multi-Agent Path Finding \sep
Warehouse Robotics \sep
Safe Interval Path Planning
\end{keywords}

\maketitle

\section{Introduction}

Automated warehouses and similar logistics systems require fleets of mobile robots to repeatedly transport items while avoiding collisions.
Multi-Agent Pickup and Delivery (MAPD) models this problem as online pickup-delivery task assignment and path planning on a graph \cite{ma2017lifelong,ma2019lifelong}.
Guaranteeing completion, meaning that every task in any finite release sequence is eventually delivered, is difficult because agents that are waiting, returning, or temporarily blocked can obstruct narrow corridors and dead ends.
Classical complete MAPD methods therefore impose structural assumptions, such as well-formedness for Token Passing (TP) \cite{ma2017lifelong,cap2015complete} or biconnectivity for Priority Inheritance with Backtracking (PIBT) \cite{okumura2019priority}.
These assumptions are useful, but they can be violated in dense warehouse layouts with single-agent-width aisles, dead-end workstations, or tree-like guidepaths.
We treat these guidepaths as application-derived graph abstractions: a narrow aisle permits no side-by-side passing, a dead-end workstation offers no through route, and a parked robot can disconnect the remaining traffic area.
This yields a concrete logistics-planning problem rather than an idealized open-grid benchmark.

Safe-haven retreat is a complementary design pattern for such constrained layouts: every committed task path is paired with a validated return path to a protected waiting location.
Prior \sharpalg reserves both a task-execution path and a retreat path to the agent's dedicated initial position using Safe Interval Path Planning (SIPP), and it can overwrite a returning agent's remaining retreat suffix only after validating a new task path and retreat \cite{hirayama2026sharp,phillips2011sipp}.
This keeps a validated fallback path available while allowing agents to accept new tasks before physically returning home.

The limitation is that the retreat target itself is fixed.
We use \emph{retreat target} for the Haven selected for a committed path.
If a retreat suffix remains in force after delivery, heading back to a static start or home location can add travel that is unnecessary when other safe waiting locations are nearby.
It is tempting to replace the fixed-Haven target with the nearest available safe location.
However, this change is not just a heuristic substitution.
Other agents plan while treating exclusive Havens as blocked vertices, and committed paths are represented in a space-time reservation table.
If Haven ownership is updated naively, another agent may plan through a vertex still physically occupied by the previous owner, or a new owner may select a Haven already reserved by another agent's committed future path.
Thus dynamic Haven selection needs an ownership-transfer protocol.

This paper presents \asharpalg (Adaptive \sharpalg), a safe dynamic Haven extension of fixed-Haven retreat planning for online MAPD with finite task releases.
Despite the pronunciation, \asharpalg is not a variant of A-star search.
Upon task assignment, \asharpalg selects a nearby \emph{available} Haven and commits a full path through pickup, delivery, and the selected Haven.
Availability checks both exclusive ownership and future reservations.
When an agent switches away from a Haven that it still occupies, a pending-release rule keeps the old Haven exclusive until the agent actually departs.
The key contribution is an ownership-transfer protocol that links online task assignment, future path reservations, and exclusive waiting-location ownership while preserving the safety structure of safe-haven retreat planning.
In the experiments, the label \sharpalg denotes the restricted fixed-Haven reimplementation used as the baseline; the full prior system is discussed only as the source of the retreat-planning mechanism.
We target graph-level MAPD instances abstracted from real warehouse constraints and maintain completion guarantees on narrow or dead-end-heavy guidepath structures.
Thus, the contribution is an application-driven planning algorithm; learning is not required for its safety guarantee, although a learned or optimization-based selector can propose Havens subject to the same availability test.

The contributions are:
\begin{itemize}
\item We propose \asharpalg, a dynamic Haven extension of safe-haven retreat planning with an availability-checked, pending-release ownership-transfer protocol.
\item We prove that, under explicit Haven structure and SIPP planning assumptions, dynamic Haven updates preserve exclusivity and reservation invariants and deliver every task in any finite release sequence.
\item We evaluate \asharpalg over 72,000 runs on 14,400 paired map--agent-count--rate--seed cases; both \sharpalg and \asharpalg achieve 100\% success, and a prespecified paired comparison finds a significant makespan improvement in 107 of 138 Haven-surplus configurations, with no significant makespan losses after Holm correction across all 138 comparisons.
\end{itemize}

\section{Background and Positioning}

\subsection{MAPD and Structural Assumptions}

MAPD extends Multi-Agent Path Finding (MAPF) by introducing online pickup-delivery tasks \cite{stern2019multi,ma2017lifelong}.
Each task has a pickup vertex, a delivery vertex, and a release time; an assigned agent must visit pickup before delivery.
The finite-release setting assumes that only finitely many tasks are released, and a planner is complete if all released tasks are delivered in finite time.

Existing guarantees often rely on graph structure.
TP is complete on well-formed MAPD instances, where there are enough non-task endpoints and paths between endpoints can avoid all other endpoints \cite{ma2017lifelong}.
PIBT provides reachability guarantees on graphs where every adjacent pair of vertices lies on a simple cycle, such as biconnected graphs \cite{okumura2019priority}.
Fujitani et al.\ extend this idea to a biconnected main area with attached trees under additional restrictions \cite{fujitani2022priority}.
We use the label PIBTTP-TA for our implementation of their PIBT with temporary priority and temporary avoidance \cite{fujitani2022priority}.

Parking and dummy-endpoint mechanisms also address fallback locations: Liu et al.\ reserve dummy paths to designated parking locations in offline MAPD, Xu et al.\ use dummy endpoints in multi-goal MAPD, and Yamauchi et al.\ choose standby nodes dynamically \cite{liu2019taskpathplanning,xu2022multigoal,yamauchi2022standby}.
These mechanisms are related precedents for waiting-location management, but their guarantee classes and transfer semantics differ from ours.
Liu et al.\ and Xu et al.\ use designated or dummy endpoints to support task planning, whereas \asharpalg transfers an agent's protected retreat target under already committed future reservations.
Standby-based methods dynamically choose temporary standby nodes to avoid indefinite waiting near congested goals.
They integrate standby-node management with Token Passing, whereas our baseline and \asharpalg commit complete SIPP paths that end at persistently protected retreat targets.
Miyashita et al.\ study distributed MAPD with asynchronous execution, occasional delays, local navigation, and flexible working endpoints \cite{miyashita2023distributed}.
Their flexible endpoints relax endpoint and execution constraints, whereas \asharpalg assumes deterministic centralized reservation-table execution and atomically transfers persistent exclusive Haven ownership under already committed future reservations.
The execution models and guarantee targets are therefore complementary rather than equivalent.
The Standby-Based Deadlock Avoidance (SBDA) method could in principle be ported to our maps, but a controlled comparison would require reimplementing its task-execution, standby-node, and coordination semantics inside the common simulator; we therefore do not present an empirical comparison that would confound those changes with Haven selection.
Here, \emph{persistent ownership} means that other agents must treat the retreat target as blocked even after the finite path reservation has ended.
In contrast to standby-node use, \asharpalg transfers that ownership while a reservation-backed commitment is active.
The update is \emph{atomic}: no other assignment or planning step may observe only one of the new path reservations and the new exclusive Haven ownership.
The central distinction is therefore not merely that the waiting location changes online, but that \asharpalg proves finite-release completion while transferring exclusive retreat-Haven ownership in a reservation-based retreat planner.
The structural assumptions above remain restrictive when task endpoints lie on narrow aisles or dead-end workstations.

Other MAPD and lifelong MAPF work emphasizes throughput and scalability in large warehouses, online arrival models, integrated task assignment and path planning, capacity constraints, or energy-aware multi-task routing \cite{wurman2008coordinating,svancara2019online,li2021lifelong,chen2021integrated,kudo2023tsp,kudo2024anytime}.
For example, throughput-oriented Rolling-Horizon Collision Resolution (RHCR) decomposes lifelong MAPF into rolling-window MAPF instances and scales to large warehouse settings \cite{li2021lifelong}.
Delay-robust MAPD, dynamic-environment MAPD, and external-agent MAPD study imperfect execution, disturbance handling, or non-communicating moving agents \cite{lodigiani2023robust,flammini2024deadlocks,bonalumi2025external}.
Those settings are orthogonal to the deterministic reservation-transfer problem studied here.
Rather than claiming throughput state of the art, this paper isolates the completion-oriented constraints of exclusive Haven retreat and asks whether its retreat target can be changed online without breaking the reservation semantics.

\subsection{Prior SHARP and the Fixed-Haven Retreat Design}
\label{sec:prior-sharp}

Prior \sharpalg introduced safe-haven retreat for application-derived multi-task warehouse operations \cite{hirayama2026sharp}.
It models richer warehouse operations, including orders with multiple pickup or delivery requirements, outbound and inbound task types, kinematic costs, and dwell times.
The part most relevant here is the safe-haven retreat mechanism: when an agent accepts a task, the planner reserves both the task path and a retreat path back to the agent's dedicated initial position.
SIPP is used to generate collision-free space-time paths under a reservation table \cite{phillips2011sipp}.
If an agent is already retreating, a new task can overwrite the remaining retreat suffix only when the new task path and subsequent retreat are validated.

Our setting is deliberately narrower.
We study standard online MAPD with one pickup and one delivery per task, finite releases, and graph-level collision constraints.
We inherit the idea of maintaining a committed path ending at a safe retreat target, but ask a different question:
can the retreat target itself be changed online while preserving safety and completeness?
The answer is nontrivial because fixed-Haven \sharpalg makes the design choice to keep every agent's initial Haven for the entire run and therefore never transfers ownership of a safe waiting location.
\asharpalg introduces exactly this missing transfer mechanism.
The contribution here is separate from the full prior system: we isolate the retreat-target mechanism, introduce dynamic Haven ownership transfer, and develop the invariants, proof obligations, and experiments needed for that transfer problem.
In short, the baseline fixes the retreat target, whereas \asharpalg keeps the reservation-backed retreat idea but allows target ownership to change only through an availability-checked transfer.

\section{Problem Setting and Haven Conditions}

We consider discrete time on an undirected graph $G=(V,E)$.
At each timestep an agent may wait or move to an adjacent vertex.
We disallow vertex collisions and edge-swap collisions.
Let $A$ be the set of agents and $\mathrm{pos}_a(t)$ the position of agent $a$ at time $t$.
We assume $A$ is nonempty.

\begin{definition}[Task]
A task $\tau$ is a tuple $\Tuple{s_\tau,g_\tau,r_\tau}$, where $s_\tau\in V$ is the pickup vertex, $g_\tau\in V$ is the delivery vertex, and $r_\tau$ is the release time.
A released task is pending until assigned, in progress until its assigned agent reaches $g_\tau$, and completed at delivery.
Tasks are assigned at most once and are not preempted or transferred before delivery.
\end{definition}

Let $H\subset V$ be a designated set of Haven candidates and let $F:=V\setminus H$ be the task-supporting core.
We write $G[F]$ for the subgraph of $G$ induced by the vertices in $F$.
The current Haven assignment at time $t$ is $\eta_t:A\to H$.
Unlike fixed-Haven planning, $\eta_t(a)$ may change at task assignment time.
Each agent $a$ also has an \emph{exclusive set} $X_t(a)\subseteq H$ containing Haven vertices that remain protected for that agent.
It always contains the current Haven $\eta_t(a)$ and, during a transfer, may temporarily also contain the occupied previous Haven.
We use candidate Haven for an eligible vertex in $H$, and protected Haven for a candidate currently in some agent's exclusive set $X_t(a)$.
A candidate can be unassigned, and it is not blocked solely because it belongs to $H$.
It becomes blocked to another agent only when protected by an exclusive set or occupied by a committed space-time reservation; otherwise, SIPP may traverse it as an ordinary vertex.
It must nevertheless pass the availability test below before an agent can select it as its current Haven.
The completeness proof below uses paths in $G[F]$ as a sufficient construction.

\begin{definition}[Haven structure conditions]
\label{def:haven-structure}
The pair $(G,H)$ satisfies the Haven structure conditions when:
\begin{enumerate}
\item $G[F]$ is connected;
\item every Haven $h\in H$ has at least one neighbor in $F$; and
\item all task endpoints lie in $F$.
\end{enumerate}
\end{definition}

We additionally assume that agents start at distinct Haven candidates, so $|A|\leq |H|$ initially.
The injectivity invariant below maintains distinct current Havens at all later timesteps.

These conditions are sufficient for the feasible path construction used in Lemma~\ref{lem:quiescence} and hence for the finite-release completeness guarantee; they are not necessary conditions for an individual instance to be solvable.
They express that the task area remains connected after removing Havens, and every Haven can be entered from that connected area.
They also separate task endpoints from exclusive waiting locations.

\section{Why Naive Dynamic Haven Switching Fails}

The fixed-Haven retreat mechanism can overwrite a retreat suffix, but the endpoint of that suffix remains the same dedicated initial position.
When the endpoint itself can change, two Haven-switching failures appear, and suffix overwrite also requires careful handling of stale self-reservations.
For this section, write $\mathcal{R}_b(t')=v$ when agent $b$ has committed to occupy vertex $v$ at time $t'$, and use $t^+$ for the phase-local state immediately after a successful commitment at time $t$.
Safe transfer requires the following three obligations:
\begin{align}
\eta_{t^+}(a)\neq h_{\mathrm{old}}\ \land\
\mathrm{pos}_a(t)=h_{\mathrm{old}}
&\;\Longrightarrow\; h_{\mathrm{old}}\in X_{t^+}(a),
\label{eq:pending-obligation}\\
\mathrm{Avail}_t(a,h)
&\;\Longrightarrow\;
\forall b\neq a,\ \forall t'\geq t:\ \mathcal{R}_b(t')\neq h,
\label{eq:reservation-obligation}\\
\mathrm{Commit}(a,\pi,h)
&\;\Longrightarrow\;
(\mathcal{R}_a^{>t},\eta_t(a),X_t(a))
\text{ are replaced as one state transition}.
\label{eq:atomic-obligation}
\end{align}
The counterexamples below respectively negate these obligations: exposing an occupied old Haven violates execution safety, selecting a future-reserved Haven violates reservation exclusion, and retaining stale self-reservations violates replacement semantics.

\begin{example}[Premature release of an occupied Haven]
Suppose an idle agent $a$ is physically located at its old Haven $h_{\mathrm{old}}$ at time $t$.
If $a$ is assigned a task and immediately releases $h_{\mathrm{old}}$ while its first planned move leaves at time $t+1$, another agent may plan through $h_{\mathrm{old}}$ at time $t$ or $t+1$.
The ownership function would say the vertex is free, but the execution state says it is occupied.
The ownership view and the execution view of the vertex disagree.
\end{example}

\begin{example}[Selecting a future-reserved Haven]
Suppose a candidate Haven $h$ is not currently owned by another agent but appears in another agent's committed future path.
If $a$ selects $h$ using ownership alone, then the new path may create a vertex-time conflict with that already committed passage.
Conversely, if the planner relies only on finite path reservations and ignores exclusive ownership, it may miss that a current Haven is persistently protected for its owner.
Safe switching therefore needs both ownership and future-reservation tests.
\end{example}

\begin{example}[Self-collision with stale reservations]
When a retreating agent accepts a new task, its own unexecuted future reservations must be overwritten.
Treating the old retreat suffix as an obstacle to the replacement plan may incorrectly reject feasible updates, while retaining stale reservations after commitment can block later agents.
\end{example}

The first two failures are specific to changing the retreat target.
\asharpalg addresses them by defining availability over both ownership and future reservations, and by delaying the release of a still-occupied old Haven until departure.
Equations~\eqref{eq:pending-obligation}--\eqref{eq:atomic-obligation} isolate the roles of the protocol components: the future-reservation test prevents selecting a candidate with another committed vertex-time use, pending release prevents premature access to an occupied old Haven, and controlled replacement of the selected agent's own future reservations prevents self-blocking.

\section{\texorpdfstring{A$\sharp$: Dynamic Haven Retreat Planning}{A-sharp: Dynamic Haven Retreat Planning}}

As introduced above, each agent has a current Haven assignment $\eta_t(a)$ and an exclusive set $X_t(a)\subseteq H$.
Here $X_t(a)$ records Haven vertices protected for $a$.
The current Haven $\eta_t(a)$ is always in $X_t(a)$.
During pending release, $X_t(a)$ may also contain the previous Haven until $a$ departs from it.
An available Haven is a candidate that passes both the ownership test induced by these exclusive sets and the future-reservation test in Definition~\ref{def:available-haven}.
Other agents treat $\bigcup_{b\neq a}X_t(b)$ as static vertex exclusions while planning; the reservation table stores finite vertex-time reservations, and directed moves induced by consecutive vertex reservations are checked to forbid edge swaps.
For an agent $a$ being replanned at time $t$, its replaceable future reservations are its own reservations at times $t'>t$.
The time-$t$ reservation is not replaceable because the replacement path must start at the agent's current reserved position.

\begin{definition}[Available Haven]
\label{def:available-haven}
A Haven candidate $h\in H$ is available for agent $a$ at time $t$ if:
\begin{enumerate}
\item $h\notin \bigcup_{b\neq a}X_t(b)$; and
\item no agent $b\neq a$ has a committed reservation occupying $h$ at any time $t'\geq t$.
\end{enumerate}
Availability is agent-relative: the tests ignore $a$'s own exclusive set and reservations.
The reservation at time $t$ is allowed because the replacement path starts at $\mathrm{pos}_a(t)$, while $a$'s reservations at times $t'>t$ are replaceable by a successful commitment.
We write $\mathrm{Avail}_t(a,h)$ for the conjunction of the two availability conditions above.
\end{definition}

\paragraph{Self-availability.}
Under the execution and reservation invariants below, agent $a$'s current Haven is available to $a$: disjoint exclusive sets prevent another owner, reservation-exclusion prevents another future reservation there, and $a$'s own waiting reservations are ignored when replanning for $a$.
Thus the available Haven set used by the algorithm is nonempty whenever the invariants hold.

\paragraph{Formal--implementation correspondence.}
The two data structures have different roles.
The exclusive sets $X_t$ persistently block protected Havens, while the reservation table stores only the finite committed path to the selected Haven.
A vertex reservation $(v,t)$ denotes occupancy of $v$ at time $t$, and consecutive reservations of the same agent induce the directed moves used for edge-swap checking.
After the finite path ends, $X_t$ continues protecting the Haven without requiring infinitely many wait reservations.

We use $t$ for the decision phase before executing the transition to $t+1$.
Within that phase, the phase-local state is the within-timestep version of $\eta_t$, $X_t$, and $\mathcal{R}$; each successful commitment updates it immediately, and later candidate evaluations see the updated state.
More precisely, write $S_t^{(k)}=(\eta_t^{(k)},X_t^{(k)},\mathcal{R}_t^{(k)})$ for the phase-local state after the $k$th successful commitment, with $S_t^{(0)}$ denoting the timestep-boundary state; unindexed state in an assignment-loop iteration means the current $S_t^{(k)}$.
The assignment loop is invoked at every timestep while pending tasks remain; an iteration with no feasible commitment leaves the pending tasks for later timesteps.
The timestep order is assignment, execution of one reserved transition, and post-execution pending-release cleanup.

An agent is task-executing until it reaches the delivery vertex of its assigned task.
After delivery, it is retreating until it reaches its current Haven.
An agent is idle when it waits at its current Haven with no assigned task.
Non-idle agents are task-executing or retreating agents with a finite committed path.
Only idle and retreating agents are eligible for a new assignment.
An agent that still holds a pending-release old Haven has not departed from that Haven after a previous assignment and is task-executing, hence it is not eligible for another assignment until the pending release is cleaned up.

When an idle or retreating agent is considered for assignment, \asharpalg chooses a pending task by nearest pickup distance and chooses a target Haven by nearest delivery-to-Haven distance among available candidates.
Here $\mathrm{dist}$ is shortest-path distance on $G$ with unit edge costs; SIPP validation, not this distance heuristic, decides feasibility under reservations and exclusive Havens.
Let $\mathrm{diam}(G[F])$ denote the diameter of the task-supporting core.
For a SIPP segment starting at absolute time $t_s$, the \emph{relative planning horizon} $T_{\max}$ permits arrivals through time $t_s+T_{\max}$.
We require $T_{\max}\geq \mathrm{diam}(G[F])+1$ in the theorem below.
The abstract helper routine \textsc{PlanFullPath}$(a,\tau,h,t,\mathcal{R},X_t,T_{\max})$ returns a Boolean success flag and, when successful, a full path $\pi$ from the current position through pickup and delivery to the selected Haven.
It concatenates three SIPP segments,
$\mathrm{pos}_a(t)\to s_\tau$, $s_\tau\to g_\tau$, and $g_\tau\to h$.
Each segment starts at the arrival time of the previous segment and may include wait actions; successful segments are concatenated after dropping the duplicated endpoint entry at each join.
Each segment treats other agents' reservations and exclusive Havens as constraints, and excludes agent $a$'s replaceable future reservations because a successful commitment atomically deletes and replaces them.
The routine is side-effect free: if any segment fails, no reservation or ownership state is changed.

If validation succeeds, \asharpalg commits by keeping $a$'s time-$t$ vertex reservation fixed, deleting $a$'s vertex reservations for times $t'>t$, inserting the future vertex entries of the new full path, treating consecutive entries of that path from time $t$ onward as the induced directed moves, updating $\eta_t(a)$ and $X_t(a)$, and setting pending release for the old Haven when needed.
After the finite path reaches the selected retreat target, the agent waits at that current Haven under the persistent protection described above.
If the selected Haven is the current Haven, the ownership update is a no-op and no pending-release Haven is created.
After the execution step, \asharpalg removes a pending old Haven from $X_{t+1}(a)$ exactly when agent $a$ has departed from that Haven.
Until cleanup, the old Haven remains blocked to others.

The \sharpalg comparison isolates the effect of dynamic Haven ownership transfer within the same one-pickup-one-delivery MAPD simulator.
In each assignment-loop pass, the variable $c^*$ stores the best validated candidate ranked by pickup distance; at most one such candidate is committed before the next pass begins.
More explicitly, a pass evaluates one greedy task--Haven candidate for every currently eligible agent against the same phase-local state, chooses one globally best validated candidate, commits it, and then recomputes all candidates in the updated state.
Algorithm~\ref{alg:asharp-workshop} gives the adaptive case; Appendix~\ref{app:fixed-pseudocode} gives the restricted fixed-Haven baseline.

During active traffic, \asharpalg evaluates only one greedy task--Haven pair per considered agent: the nearest-pickup pending task and the nearest available Haven for that task.
If validation fails, assignment may be delayed even when another pending task or a farther available Haven would validate.
This restriction is a throughput heuristic, not a safety condition.
Segment-wise validation is a sound feasibility check for the committed path, but during active traffic it is not a global completeness claim over all possible intermediate arrival times and Haven choices.

\begin{algorithm}[!htbp]
\caption{\asharpalg assignment loop}
\label{alg:asharp-workshop}
\footnotesize
\begin{algorithmic}[1]
\Require Decision time $t$; pending tasks $Q$; agents $A$; Haven candidates $H$; assignments $\eta_t$; exclusive sets $X_t$; reservation table $\mathcal{R}$; segment horizon $T_{\max}$
\Ensure Updated $Q$, $\eta_t$, $X_t$, $\mathcal{R}$, and agent commitments
\State Add tasks released at time $t$ to the pending set $Q$
\State $I \gets$ idle or retreating agents with no pending-release old Haven
\While{$Q\neq \emptyset$ and $I\neq \emptyset$}
  \State $c^*\gets \bot$; $d^*\gets \infty$
  \For{each $a\in I$}
    \State $\tau\gets \arg\min_{\tau'\in Q}\mathrm{dist}(\mathrm{pos}_a(t),s_{\tau'})$
    \State $\mathcal{H}_a(t)\gets\{h'\in H:\ \mathrm{Avail}_t(a,h')\}$
    \State $h\gets \arg\min_{h'\in\mathcal{H}_a(t)}\mathrm{dist}(g_\tau,h')$
    \State $(\mathit{success},\pi)\gets \textsc{PlanFullPath}(a,\tau,h,t,\mathcal{R},X_t,T_{\max})$
    \If{$\mathit{success}\ \mathrm{and}\ \mathrm{dist}(\mathrm{pos}_a(t),s_\tau)<d^*$}
      \State $c^*\gets(a,\tau,h,\pi)$; $d^*\gets\mathrm{dist}(\mathrm{pos}_a(t),s_\tau)$
    \EndIf
  \EndFor
  \If{$c^*=\bot$}
    \State \textbf{break}
  \EndIf
  \State Let $(a,\tau,h,\pi)\gets c^*$ and let $h_{\mathrm{old}}\gets \eta_t(a)$
  \State Keep $a$'s vertex reservation at time $t$ fixed
  \State Delete $a$'s vertex reservations for $t'>t$ and insert the future vertex entries of $\pi$ in $\mathcal{R}$
  \State Treat consecutive entries of $\pi$ from time $t$ onward as $a$'s committed directed moves
  \State Assign $\tau$ to $a$ and mark $a$ task-executing
  \If{$h\neq h_{\mathrm{old}}$}
    \State $\eta_t(a)\gets h$
    \State $X_t(a)\gets \{h\}$
    \If{$\mathrm{pos}_a(t)=h_{\mathrm{old}}$}
      \State $X_t(a)\gets X_t(a)\cup\{h_{\mathrm{old}}\}$
    \EndIf
  \Else
    \State Leave $\eta_t(a)$ and $X_t(a)$ unchanged because $h=h_{\mathrm{old}}$
  \EndIf
  \State Remove the selected agent and task from $I$ and $Q$
\EndWhile
\State Execute one timestep along committed reservations
\State After delivery mark an agent retreating; at its current Haven with no task mark it idle
\State Release $h_{\mathrm{old}}$ when $\mathrm{pos}_a(t+1)\neq h_{\mathrm{old}}$; then renew rolling wait reservations for agents with no task
\end{algorithmic}
\end{algorithm}

One assignment-loop pass scans the pending task set $Q$ and the Haven candidates for each agent in the eligible set $I$, and performs at most one full-path SIPP validation per agent.
Here $|\cdot|$ denotes set cardinality.
Assuming shortest-path distance fields used by the greedy selectors are precomputed or cached, this costs $O(|I|(|Q|+|H|C_{\mathrm{avail}}+C_{\mathrm{SIPP}}))$ per pass, where $C_{\mathrm{avail}}$ is the cost of one availability test and $C_{\mathrm{SIPP}}$ is the cost of the three SIPP segment searches.
A timestep can contain up to $\min(|I|,|Q|)$ successful passes until either eligible agents or pending tasks are exhausted.
Ignoring the shrinkage of $I$ and $Q$ across successful passes, this gives a worst-case per-timestep cost of $O(\min(|I|,|Q|)|I|(|Q|+|H|C_{\mathrm{avail}}+C_{\mathrm{SIPP}}))$.
A timestep may also include one final failed pass before the loop breaks; this does not change the asymptotic bound.

Because the current Haven is available to its owner under the invariants, $\mathcal{H}_a(t)$ is nonempty.
If no alternative is available, \asharpalg selects its current Haven and performs no ownership transfer for that assignment.
When current Havens are initialized to the baseline's fixed Havens and no ownership transfer has occurred, this coincides with the fixed-Haven \sharpalg choice.
All loops and argmin operations use deterministic task, agent, and Haven orders; equal-distance best candidates retain the first item in these orders.
The experimental protocol below specifies the orders.

\section{Theoretical Guarantees}

We state the proof for the finite-release MAPD setting above, where each task has one pickup and one delivery.
Prior \sharpalg provides the motivating safe-haven retreat mechanism, but the dynamic Haven invariant and theorem below are specific to \asharpalg.
Throughout this section, invariants are evaluated at the beginning of a timestep, after the previous execution and pending-release cleanup have finished.
The guarantee applies to graph-level finite-release MAPD instances with deterministic discrete-time execution, a centralized reservation table, distinct initial Havens, non-preemptive one-pickup-one-delivery tasks, maps satisfying Definition~\ref{def:haven-structure}, and SIPP soundness for all validation calls.
The assignment loop is invoked at every timestep while pending tasks remain and considers only idle or retreating agents with no pending-release old Haven.
SIPP validation is assumed sound with respect to static vertex exclusions induced by exclusive Haven sets, finite vertex-time reservations, and edge-swap constraints induced by consecutive reservations.
Its completeness assumption is limited to the quiescent segment subproblems used in Lemma~\ref{lem:quiescence} within each per-segment horizon $T_{\max}\geq\mathrm{diam}(G[F])+1$ under the same static exclusions and finite reservation constraints.
It does not model execution delays, localization errors, dynamic obstacles, or shared parking.
Committed future path reservations are respected during validation, while exclusive Haven sets provide persistent static exclusions for protected waiting locations.
The proof structure is as follows: Proposition~\ref{prop:update} handles ownership updates, Lemma~\ref{lem:invariants} preserves safety and reservation invariants for one timestep, Lemma~\ref{lem:quiescence} gives progress once the system is quiescent, and Theorem~\ref{thm:complete} combines finite releases with repeated assignment-loop invocation.

\subsection{Invariants}

\begin{definition}[Execution invariants]
At every timestep:
\begin{enumerate}
\item $\eta_t$ is injective;
\item for every agent $a$, $\eta_t(a)\in X_t(a)$;
\item for any $a\neq b$, $X_t(a)\cap X_t(b)=\emptyset$; and
\item $X_t(a)\setminus\{\eta_t(a)\}$ contains at most one pending-release old Haven, which is removed by the post-execution cleanup immediately after $a$ departs from it.
\end{enumerate}
\end{definition}

\begin{definition}[Reservation invariants]
At every timestep:
\begin{enumerate}
\item at most one agent reserves any vertex-time pair;
\item no committed directed moves form an edge swap;
\item for every agent $a$, $\mathrm{pos}_a(t)$ matches the vertex reserved for $a$ at time $t$, and each executed transition follows the corresponding reserved transition, which is either a wait or an edge in $G$;
\item every non-idle agent has a finite committed path ending at its current Haven;
\item idle agents wait at their current Havens, which are persistently protected by their exclusive sets; and
\item (reservation-exclusion) for any agents $a\neq b$ and any $h\in X_t(a)$, agent $b$ has no committed reservation occupying $h$ at any time $t'\geq t$.
\end{enumerate}
\end{definition}

\begin{proposition}[Dynamic Haven update preserves exclusive-set invariants]
\label{prop:update}
Assume the execution invariants hold before an assignment at time $t$.
Assume also that the selected eligible agent $a$ has no pending-release old Haven before this assignment.
If $a$ selects an available Haven $h_{\mathrm{new}}$ and applies the ownership update in Algorithm~\ref{alg:asharp-workshop}, then the execution invariants continue to hold.
\end{proposition}

The proof separates the no-op case $h_{\mathrm{new}}=h_{\mathrm{old}}$ from a transfer and uses availability, pending release, and the eligibility rule to preserve injectivity, disjointness, and the one-pending-Haven bound.
Full details appear in Appendix~\ref{app:proofs}.

\begin{lemma}[One-timestep invariant preservation]
\label{lem:invariants}
If the execution and reservation invariants hold at the beginning of timestep $t$, then after Algorithm~\ref{alg:asharp-workshop} completes its assignment phase, executes one reserved transition, and performs pending-release cleanup, the invariants hold at the beginning of timestep $t+1$.
\end{lemma}

The proof applies Proposition~\ref{prop:update} and planner soundness to each atomic commitment in phase-local order, then observes that reserved execution and post-departure cleanup preserve the resulting conditions at $t+1$.
Full details appear in Appendix~\ref{app:proofs}.

\subsection{Finite-Release Completeness}

\begin{definition}[Quiescent configuration]
The system is quiescent at time $T$ if every agent is idle at its current Haven, no finite committed movement suffix remains, and the only persistent Haven constraints are the exclusive sets.
An implementation may materialize a bounded number of wait reservations at an idle agent's Haven and renew them as time advances; these are \emph{rolling wait entries}.
They do not represent a movement suffix, are ignored when that same agent replans, and therefore do not change the quiescence condition.
\end{definition}

The next lemma does not rely on the nearest-Haven heuristic; it shows that, once traffic has quiesced, every available Haven gives a feasible full path.

\begin{lemma}[Planning succeeds in quiescence]
\label{lem:quiescence}
Assume the Haven structure conditions and SIPP completeness for the quiescent segment subproblems, under the static vertex exclusions induced by exclusive Haven sets and the finite reservation constraints, within a horizon at least $\mathrm{diam}(G[F])+1$ for each segment.
In a quiescent configuration, for any pending task $\tau$, any agent $a$, and any Haven $h$ available for $a$, \textsc{PlanFullPath} succeeds for the path $\mathrm{pos}_a(T)\to s_\tau\to g_\tau\to h$ under the reservation and exclusivity constraints stated above.
\end{lemma}

The proof constructs all three segments through the connected core $G[F]$, using one additional edge to leave the start Haven and one to enter the target Haven; availability and quiescence remove conflicting ownership and movement reservations.
Full details appear in Appendix~\ref{app:proofs}.

\begin{theorem}[Finite-release completeness]
\label{thm:complete}
Assume the Haven structure conditions, distinct initial Havens, initial execution and reservation invariants, deterministic execution of committed reservations, non-preemptive tasks assigned at most once, finitely many task releases, and assignment-loop invocation at every timestep while pending tasks remain using the eligibility rule in Algorithm~\ref{alg:asharp-workshop}.
Assume also SIPP soundness for all validation calls and SIPP completeness for the quiescent segment subproblems of Lemma~\ref{lem:quiescence} under static Haven exclusions and finite reservation constraints within each per-segment horizon $T_{\max}$ with $T_{\max}\geq\mathrm{diam}(G[F])+1$.
Then \asharpalg delivers every released task in finite time.
\end{theorem}

The proof assumes an unfinished task persists, takes a time after the last release and last successful commitment, and lets all finite committed suffixes reach quiescence.
Because task-executing agents are ineligible for reassignment, an assigned task's pickup--delivery prefix is not replaced before delivery.
Lemma~\ref{lem:quiescence} then forces another successful commitment, contradicting the choice of that time.
Full details appear in Appendix~\ref{app:proofs}.

\begin{corollary}[Fixed-Haven case]
If \asharpalg always selects the current Haven and current Havens are initialized to the baseline's fixed Havens, Haven ownership never changes.
The algorithm then recovers the restricted fixed-Haven baseline used in our experiments, and Theorem~\ref{thm:complete} applies under the same assumptions.
\end{corollary}

\section{Experiments}

\subsection{Maps and Setup}

We evaluate \asharpalg on four maps with different structural properties (Figure~\ref{fig:maps}).
The well-formed map is the public MAPD benchmark from \cite{ma2017lifelong}, reflecting Kiva-style warehouse layouts \cite{wurman2008coordinating}.
The constrained maps are narrow-biconnected (narrow-bi), narrow-biconnected with depth-1 dead ends (narrow-bi-dead), and a tree-like task-area layout; the latter two include structures excluded by well-formedness or biconnectivity assumptions.
The tree map represents a tree-structured warehouse guidepath with narrow aisles and branches that provide access to workstations and storage locations.
Such patterns arise in space-efficient automated warehouses and are studied in prior warehouse-layout and MAPD work \cite{iida2023negotiation,hirayama2025taai,azadeh2017robotized,roy2017multi}.
All maps were checked against the Haven structure conditions before evaluation: the core is connected, every Haven is adjacent to the core, and task endpoints are sampled from the core.

\begin{figure}[!htbp]
\centering
\begin{subfigure}{0.47\linewidth}
\centering
\includegraphics[width=\linewidth]{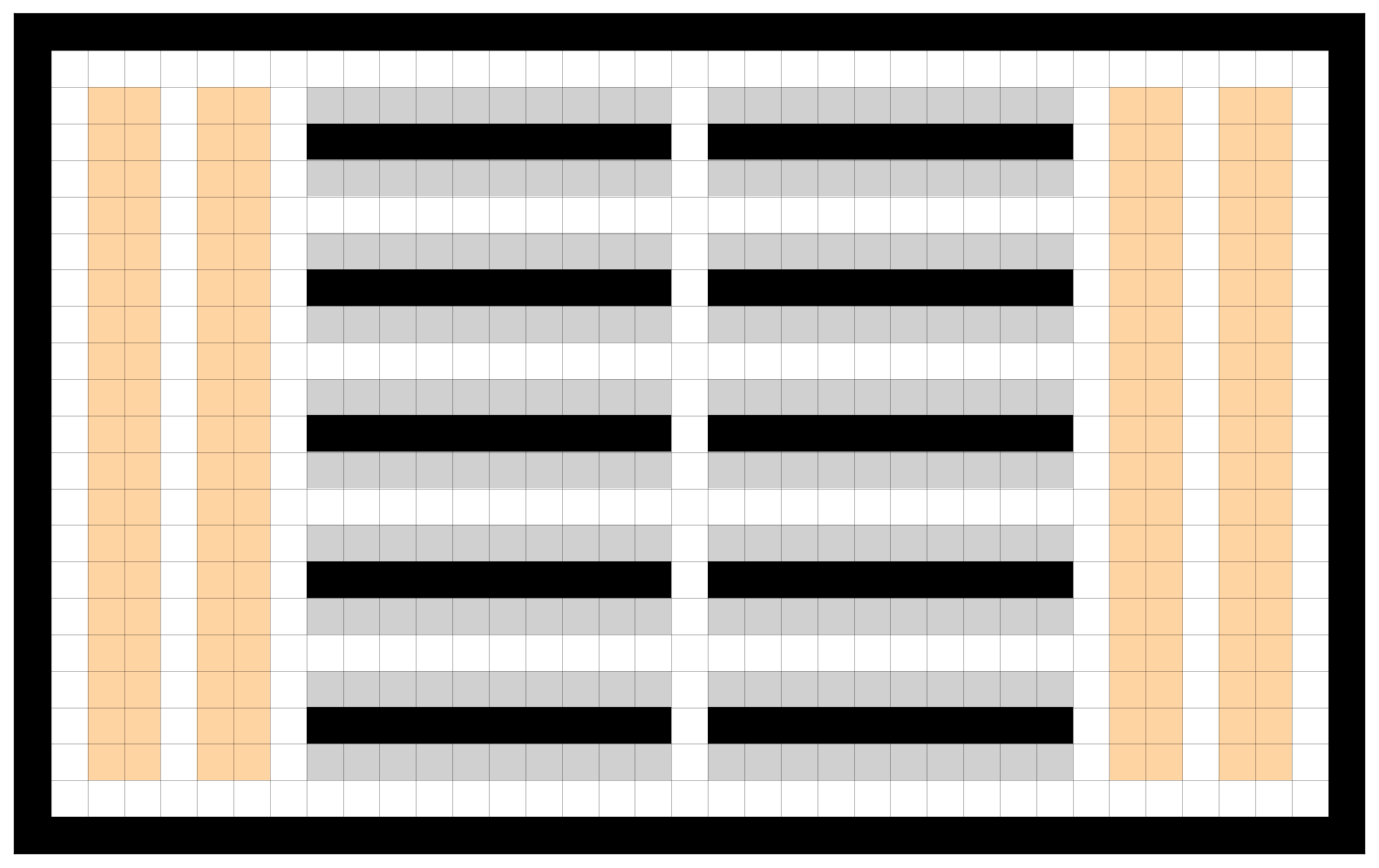}
\caption{well-formed}
\end{subfigure}
\hfill
\begin{subfigure}{0.47\linewidth}
\centering
\includegraphics[width=\linewidth]{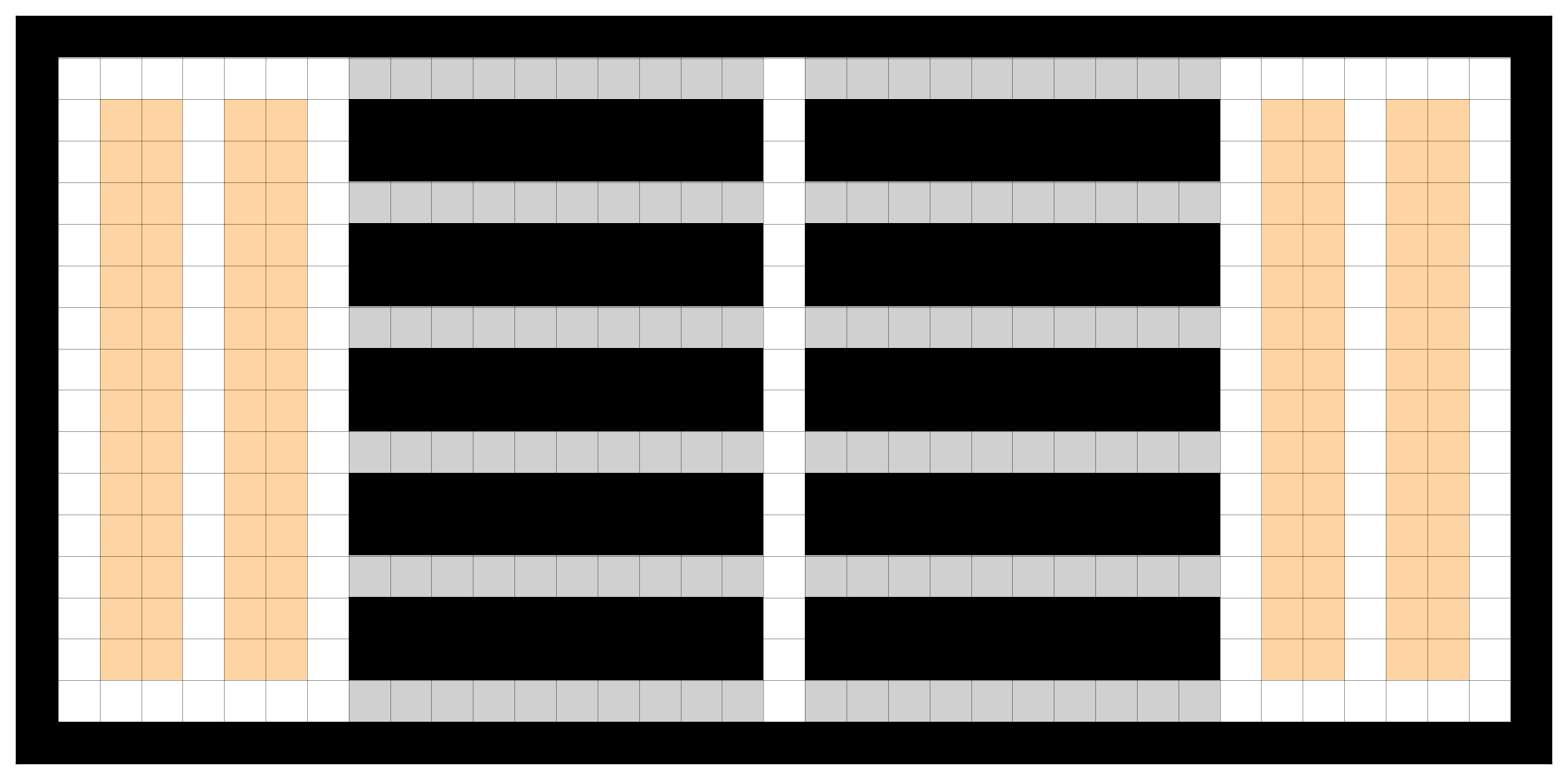}
\caption{narrow-bi}
\end{subfigure}

\begin{subfigure}{0.47\linewidth}
\centering
\includegraphics[width=\linewidth]{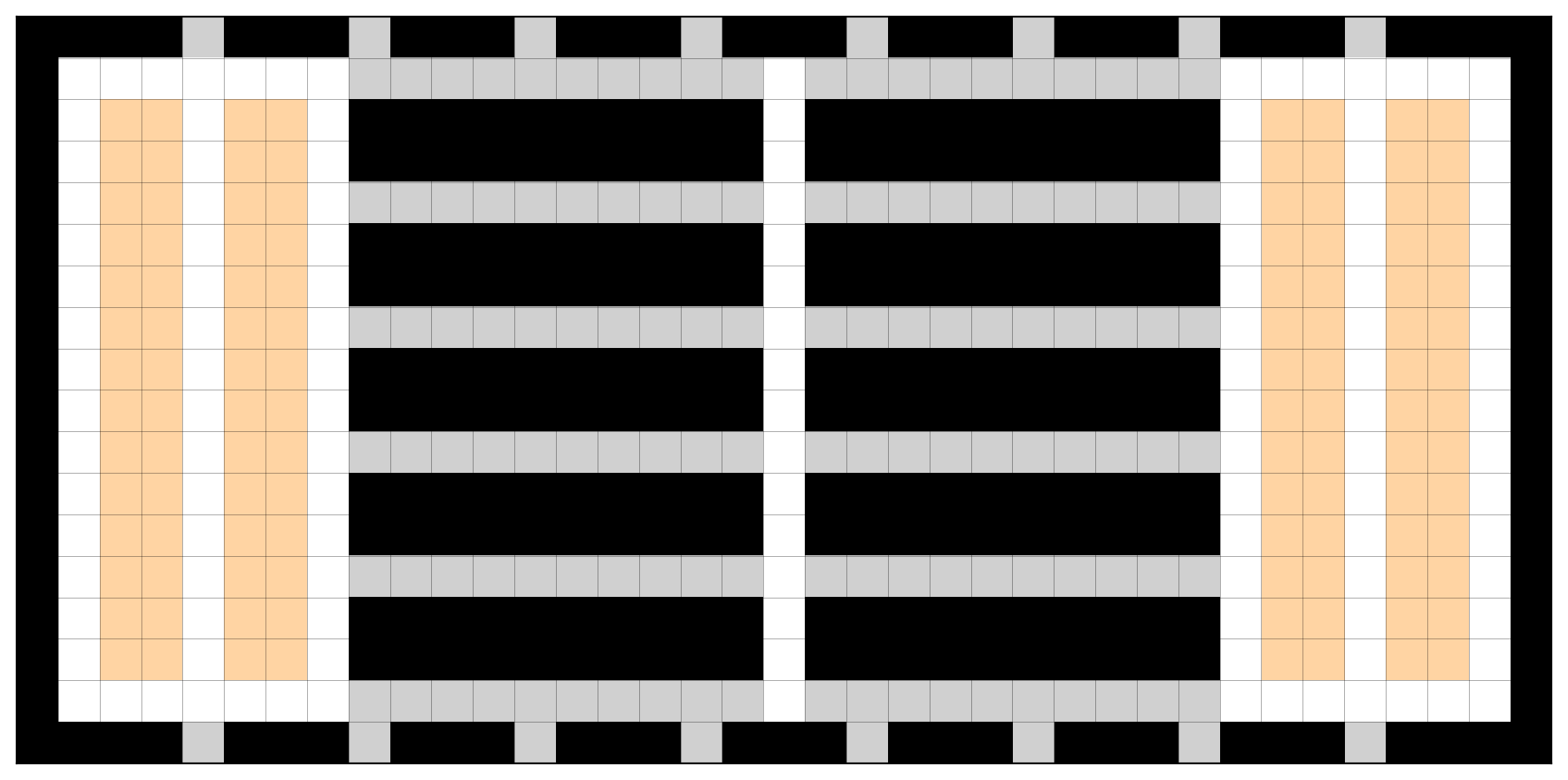}
\caption{narrow-bi-dead}
\end{subfigure}
\hfill
\begin{subfigure}{0.47\linewidth}
\centering
\includegraphics[width=\linewidth]{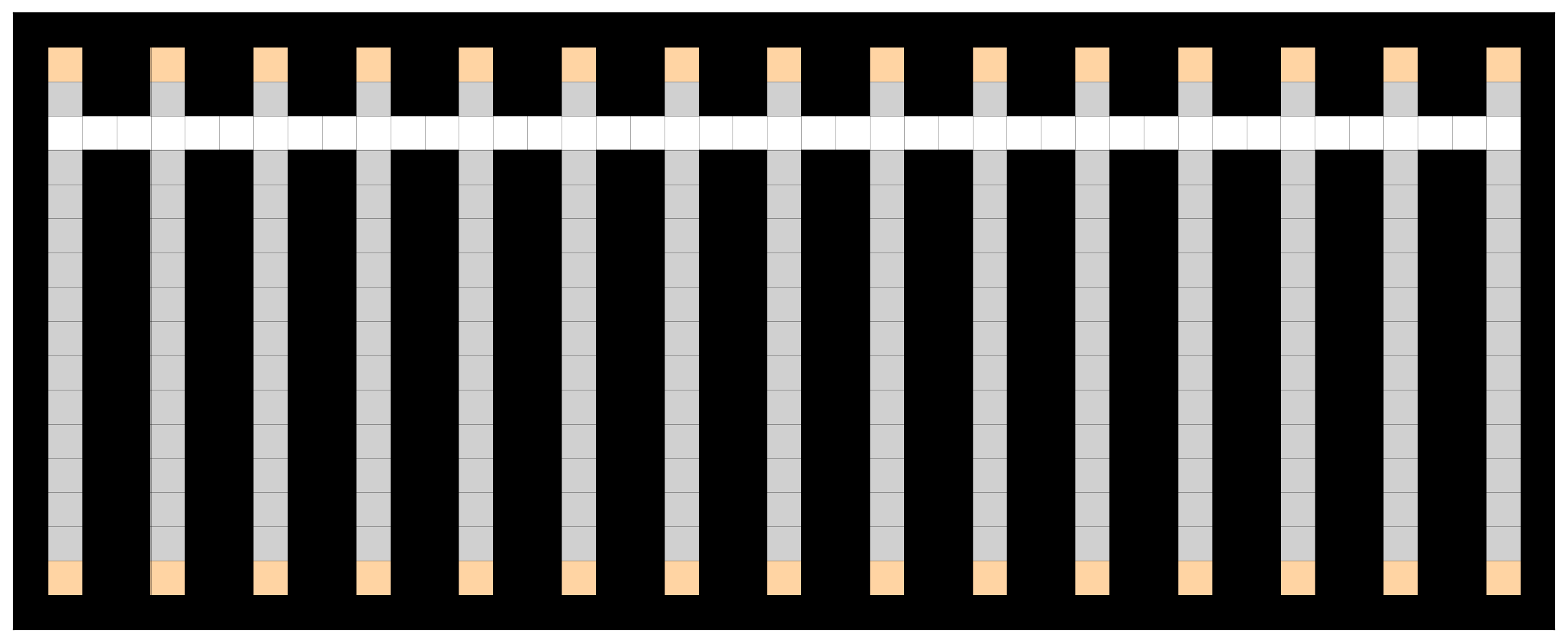}
\caption{tree}
\end{subfigure}
\caption{Maps used in the evaluation. Orange cells are Haven candidates, gray cells are task endpoint candidates, and black cells are obstacles. The well-formed map is the public benchmark, narrow-biconnected (narrow-bi) keeps biconnectivity with single-width aisles, narrow-biconnected with depth-1 dead ends (narrow-bi-dead) adds dead-end workstations, and tree violates biconnectivity with a tree-like task area.}
\label{fig:maps}
\end{figure}

\begin{table}[!htbp]
\caption{Map statistics. $P_{\mathrm{end}}$ denotes the endpoint-candidate set.}
\label{tab:map-stats}
\centering
\small
\begin{tabular}{lrrrrrrr}
\toprule
Map & Size & $|V|$ & $|F|$ & $|H|$ & $|P_{\mathrm{end}}|$ & Obstacles & $\mathrm{diam}(G[F])$ \\
\midrule
well-formed & $37{\times}23$ & 635 & 483 & 152 & 200 & 216 & 54 \\
narrow-bi & $37{\times}18$ & 360 & 248 & 112 & 120 & 306 & 49 \\
narrow-bi-dead & $37{\times}18$ & 376 & 264 & 112 & 136 & 290 & 49 \\
tree & $45{\times}18$ & 268 & 238 & 30 & 195 & 542 & 66 \\
\bottomrule
\end{tabular}
\end{table}

We vary $|A|\in\{5,10,15,20,25,30\}$ and task generation rate $\lambda\in\{0.5,1.0,1.5,2.0,2.5,3.0\}$.
For each map, task rate, and seed, we generate one schedule and replay it for every agent count and algorithm.
The resulting 2,400 unique schedules define 14,400 paired map--agent-count--rate--seed cases and 72,000 runs across TP, PIBT, PIBTTP-TA, \sharpalg, and \asharpalg.
Within each such case, all five algorithms therefore receive the same pre-generated task schedule.
For each map, agent count, and seed $s\in\{1001,\ldots,1100\}$, the implementation seeds Python's pseudorandom-number generator with $s$, sorts Haven coordinates in row-major order, samples $|A|$ distinct Havens uniformly without replacement, and assigns the sampled sequence to agents in increasing ID order.
Every algorithm uses the same initial agent--Haven mapping in a paired case; consequently, the 100 seeds vary both the initial fixed-Haven assignment and the task schedule, although these two factors are not independently crossed.
For each timestep $t\in\{0,\ldots,300\}$, the number of released tasks is $\lfloor\lambda\rfloor+\mathrm{Bernoulli}(\lambda-\lfloor\lambda\rfloor)$.
Because the maps in Table~\ref{tab:map-stats} have no explicit delivery-only cells, each task samples pickup and delivery as two distinct vertices from this endpoint-candidate set; sampling is with replacement across different tasks.

The \sharpalg baseline and \asharpalg use the same nearest-pickup task selection and path-validation rules.
The evaluated \asharpalg instantiation adds nearest-available-Haven selection and the availability-checked, pending-release ownership transfer, whereas \sharpalg retains its initial Haven.
The protocol's safety and completeness results are selector-independent under the stated assumptions, but the empirical efficiency results are specific to this nearest-Haven selector.
For both safe-haven retreat variants, the per-segment SIPP horizon is set to $T_{\max}=\mathrm{diam}(G[F])+1$ on each map; from Table~\ref{tab:map-stats}, the largest value is 67 on the tree map.
Ties are deterministic: pre-generated task schedules are ordered by release time and task ID, nearest-task ties follow task iteration order, equal-distance best feasible candidates retain the first agent in increasing agent ID order, and ties among available Havens in \asharpalg use lexicographic grid-coordinate order.

TP and PIBT are based on publicly available implementations, and PIBTTP-TA is reimplemented from \cite{fujitani2022priority}; these baselines use their own assignment and movement rules, so they serve primarily to expose the effect of structural assumptions on the same paired task streams.
Thus their role is diagnostic: they expose where standard structural assumptions fail, while the main efficiency comparison is \sharpalg versus \asharpalg.
TP, PIBT, and PIBTTP-TA respectively represent the well-formed-instance, simple-cycle, and biconnected-main-area-with-attached-trees assumptions.
All five algorithms use the same collision, stall, timeout, and task-completion criteria.

A run succeeds only if all released tasks are delivered within 10,000 timesteps.
We mark a run as failed if it has a vertex collision, edge-swap collision, invalid non-adjacent move, 3,000-second per-run wall-clock timeout, or 1,000-timestep stall without task progress.
These safeguards are treated as operational failures in the simulation.
We report success rate as the fraction of runs that complete all released tasks, makespan as the final delivery timestep rather than the time for all agents to return to Havens, service time as the average release-to-delivery duration over completed tasks, and computation time as simulator wall-clock milliseconds per step.
Absolute values reported with $\pm$ are means plus or minus one standard deviation over 100 paired seeds.
The retreat-to-Haven suffix remains part of each committed path; final-delivery makespan is used only as the service-performance metric.

Under an industrial reporting constraint, we do not tabulate the author-measured absolute makespan and service-time values for the constrained maps; absolute values are reported for the public well-formed benchmark, and computation times are reported for all maps.
For the constrained maps, we report success rates, paired relative improvements, normalized trends, and statistical outcomes without disclosing normalization anchors in absolute units.
This reporting restriction does not prevent independent reproduction: the released implementation, machine-readable maps, task generator, seed protocol, and experiment configuration allow readers to rerun the experiments and obtain their own measurements.
Paired performance comparisons are made only among algorithms with 100\% success in a configuration, avoiding survivorship bias.
For statistical testing, samples are paired by random seed within each map, agent count, task generation rate, and metric.
The prespecified primary comparison is \sharpalg versus \asharpalg on the 138 Haven-surplus configurations, defined by $|A|<|H|$.
For makespan and service time separately, we apply two-sided paired Wilcoxon signed-rank tests and Holm correction across all 138 configuration-level tests; zero differences are omitted, no continuity correction is used, and SciPy~1.11.3 with \texttt{method=auto} uses its asymptotic calculation for these 100-pair tests.
We use significance level $\alpha=0.05$ throughout; comparisons involving TP, PIBT, and PIBTTP-TA are descriptive diagnostics rather than members of the primary testing family.

\subsection{Success and Task Performance}

Figure~\ref{fig:success} shows the success-rate heatmap.
TP, PIBT, and PIBTTP-TA are included as structural-assumption baselines; the heatmap should be read primarily as a stress test of their graph assumptions.
The figure is intended to be read categorically: green cells denote completed configurations, while red or orange cells identify configurations where structural assumptions or operational safeguards fail in this diagnostic benchmark.
The \sharpalg baseline and \asharpalg both achieve 100\% success in every configuration, supporting the claim that dynamic Haven updates preserve the robustness of safe-haven retreat.

\begin{figure}[!htbp]
\centering
\includegraphics[width=\linewidth]{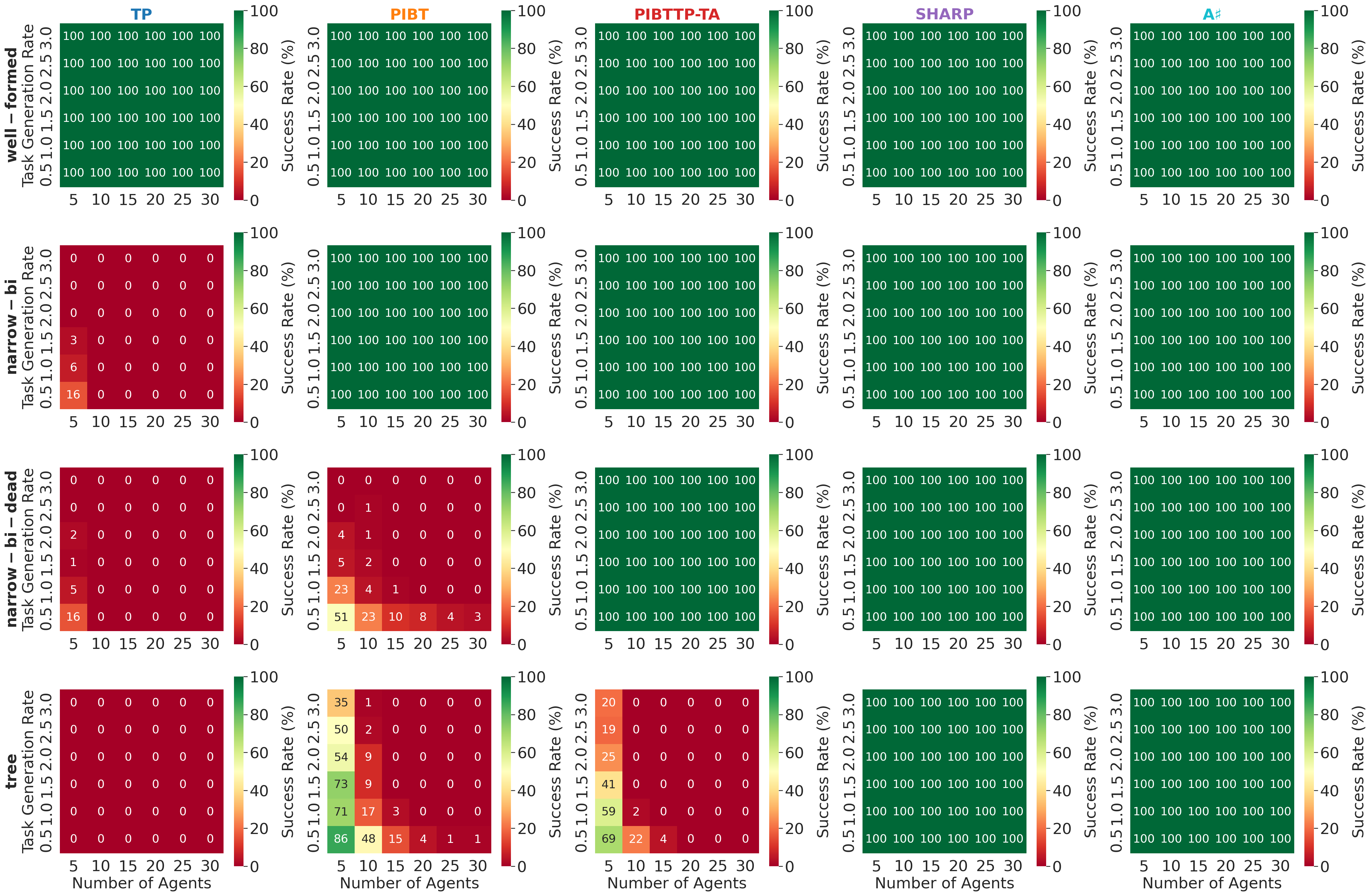}
\caption{Success rate across the four maps, agent counts, and task generation rates. Cell values give success rates, and color only highlights high versus low rates. TP, PIBT, and PIBTTP-TA serve as structural-assumption diagnostics; the main efficiency comparison is \sharpalg versus \asharpalg. Only \sharpalg and \asharpalg achieve 100\% success in every tested configuration.}
\label{fig:success}
\end{figure}

Figures~\ref{fig:service-performance} and~\ref{fig:makespan-performance} summarize service-time and makespan trends.
Values are normalized per map and metric by an endpoint anchor \sharpalg mean: the task-rate sweeps use the high-rate endpoint ($|A|=30,\lambda=3.0$), and the agent-count sweeps use the left endpoint at the same task rate ($|A|=5,\lambda=3.0$).
These endpoint anchors keep each sweep normalized within the same map and metric, rather than mixing task-rate and agent-count effects.
For each algorithm, we plot only configurations in which it completes all 100 paired seeds, matching \sharpalg; failed configurations are omitted from these trend plots.
For constrained maps, normalized values are shown without reporting the absolute \sharpalg anchors in the paper.
On the public well-formed map with 30 agents and $\lambda=3.0$, \sharpalg obtains makespan $695.2\pm16.8$ and service time $174.8\pm5.7$, while \asharpalg obtains makespan $682.6\pm15.7$ and service time $173.8\pm6.3$.
The improvement is modest on the public well-formed map, as expected, because fixed Havens are less harmful in this layout (1.8\% makespan and 0.6\% service-time reduction in this setting).
The effect is larger in constrained layouts: on the tree map with 20 agents and $\lambda=3.0$, \asharpalg reduces makespan and service time by 20.4\% and 16.2\%, respectively.

Table~\ref{tab:effect-summary} breaks down the primary \sharpalg/\asharpalg comparison for Haven-surplus configurations ($|A|<|H|$).
Haven surplus guarantees that at least one candidate is not a current Haven, but not that an alternative is available at every decision because ownership and future reservations can temporarily exclude candidates.
Configurations with $|A|=|H|$ are excluded because \asharpalg then makes the same retreat-target choice as \sharpalg in the tested setup.
The map/metric-level median $\Delta$ is the primary effect-size summary.
Across the two metric-wise families, 199 of the 276 map--metric--configuration pairs show a significant difference after Holm correction within their 138-test family; 175 favor \asharpalg and 24 favor \sharpalg.
The significant \sharpalg wins occur only for service time on the two narrow-biconnected maps, where nearest-Haven selection can create local congestion even though it shortens retreat targets.
\asharpalg is therefore not uniformly better for service time: in several narrow-biconnected configurations, the greedy nearest-Haven selector favors the fixed-Haven \sharpalg baseline.
These losses isolate a selector-level trade-off rather than a failure of the ownership-transfer protocol.
On the tree map, \asharpalg is significantly better in all 30 configurations with $|A|<|H|$ for both metrics.
When $|A|=|H|$ on the tree map, no alternative Haven is available and \asharpalg uses the same retreat-target choice as \sharpalg.

\begin{table}[!htbp]
\caption{Primary \sharpalg/\asharpalg comparison on Haven-surplus configurations ($|A|<|H|$). $\Delta=100(\mathrm{fixed}-\mathrm{adaptive})/\mathrm{fixed}$, so positive values favor \asharpalg. Only relative effects are reported. The last column counts significant \asharpalg wins, significant \sharpalg wins, and non-significant pairs after Holm correction across all 138 tests separately for each metric.}
\label{tab:effect-summary}
\centering
\small
\begin{tabular}{llrrc}
\toprule
Map & Metric & Configs & Median $\Delta$ & Sig.\ \asharpalg / \sharpalg / not sig. \\
\midrule
well-formed & makespan & 36 & +1.5\% & 33 / 0 / 3 \\
well-formed & service time & 36 & +1.0\% & 21 / 0 / 15 \\
narrow-bi & makespan & 36 & +1.0\% & 17 / 0 / 19 \\
narrow-bi & service time & 36 & -0.7\% & 8 / 15 / 13 \\
narrow-bi-dead & makespan & 36 & +1.6\% & 27 / 0 / 9 \\
narrow-bi-dead & service time & 36 & 0.0\% & 9 / 9 / 18 \\
tree & makespan & 30 & +16.7\% & 30 / 0 / 0 \\
tree & service time & 30 & +15.3\% & 30 / 0 / 0 \\
\bottomrule
\end{tabular}
\end{table}

\begin{figure}[!htbp]
\centering
\begin{subfigure}{0.48\textwidth}
\centering
\includegraphics[width=\linewidth]{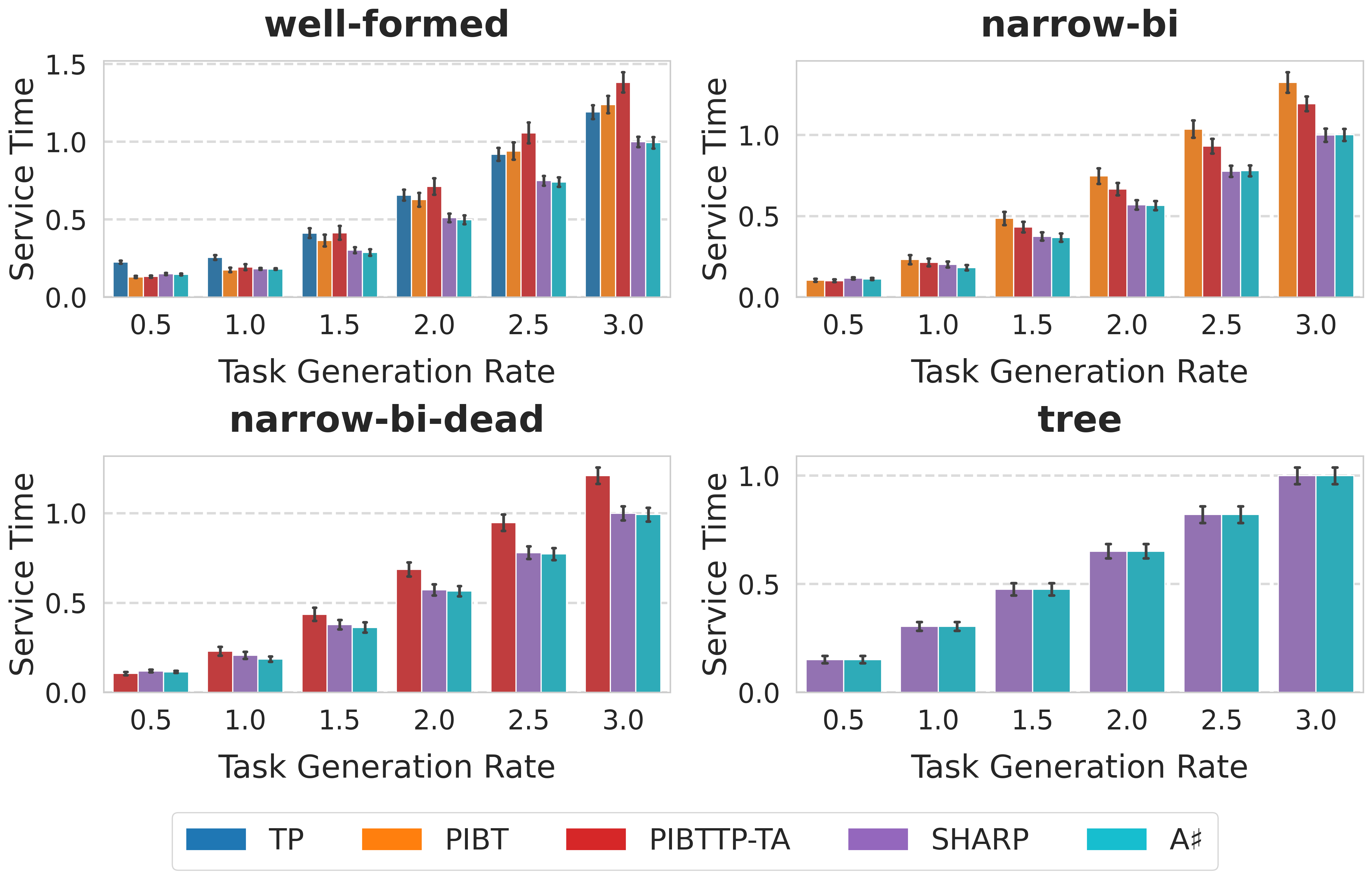}
\caption{Service time vs.\ task generation rate (30 agents).}
\end{subfigure}
\hfill
\begin{subfigure}{0.48\textwidth}
\centering
\includegraphics[width=\linewidth]{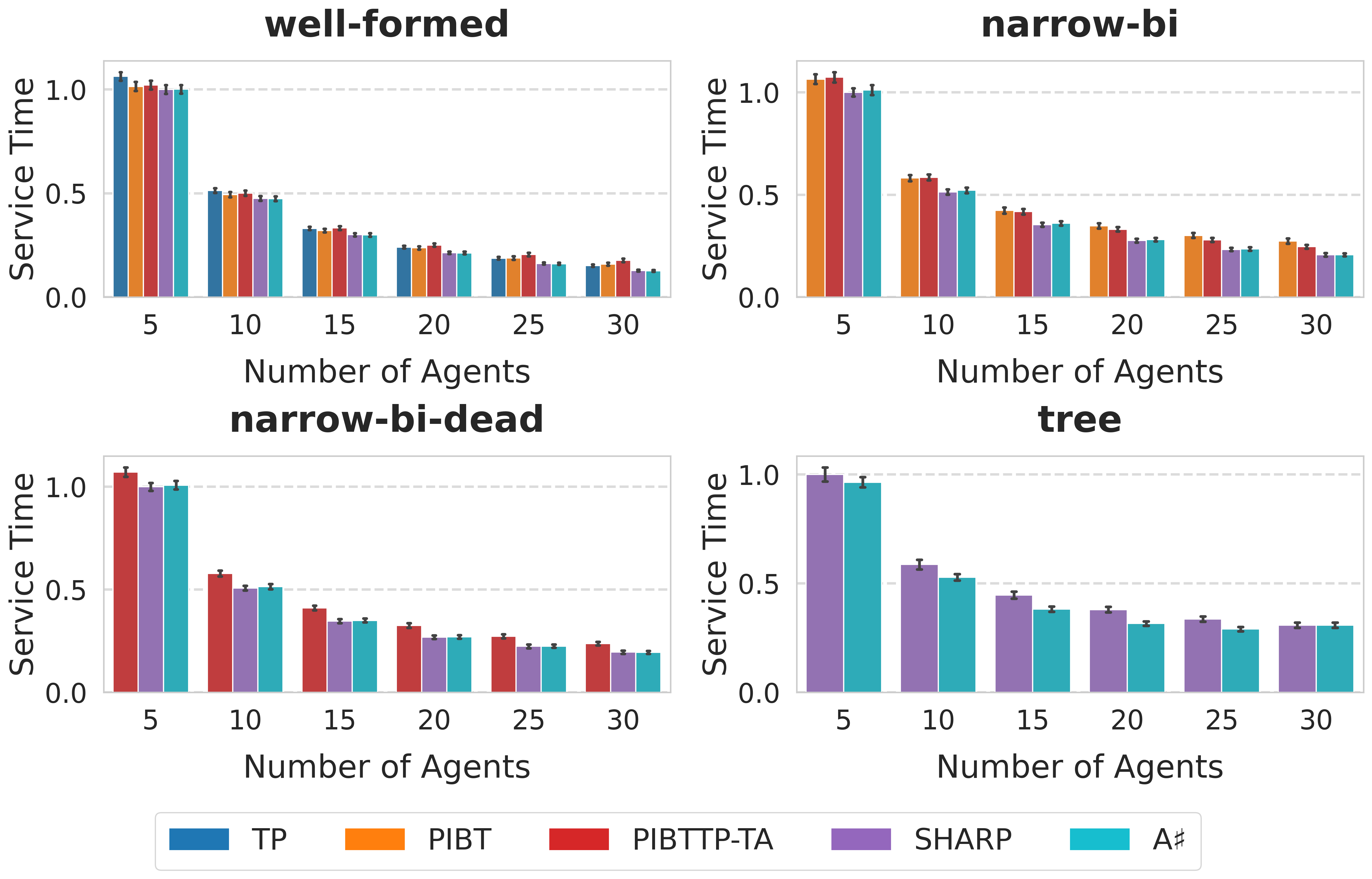}
\caption{Service time vs.\ agents ($\lambda=3.0$).}
\end{subfigure}
\caption{Normalized service time, divided by the \sharpalg anchor mean for each map and aggregation axis. Lower is better. Error bars denote one standard deviation over paired seeds; absolute anchors for constrained maps are not reported in the paper.}
\label{fig:service-performance}
\end{figure}

\begin{figure}[!htbp]
\centering
\begin{subfigure}{0.48\textwidth}
\centering
\includegraphics[width=\linewidth]{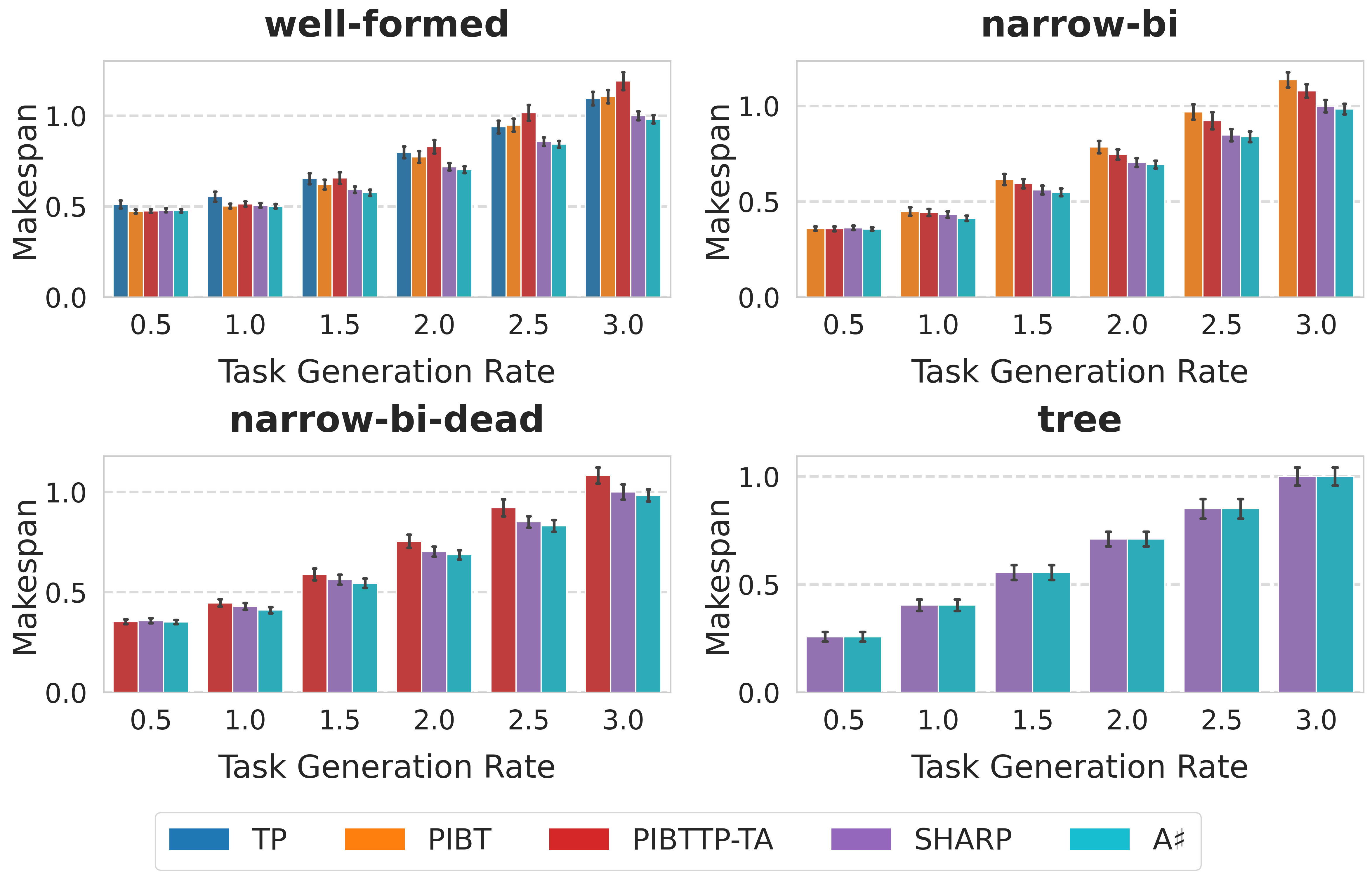}
\caption{Makespan vs.\ task generation rate (30 agents).}
\end{subfigure}
\hfill
\begin{subfigure}{0.48\textwidth}
\centering
\includegraphics[width=\linewidth]{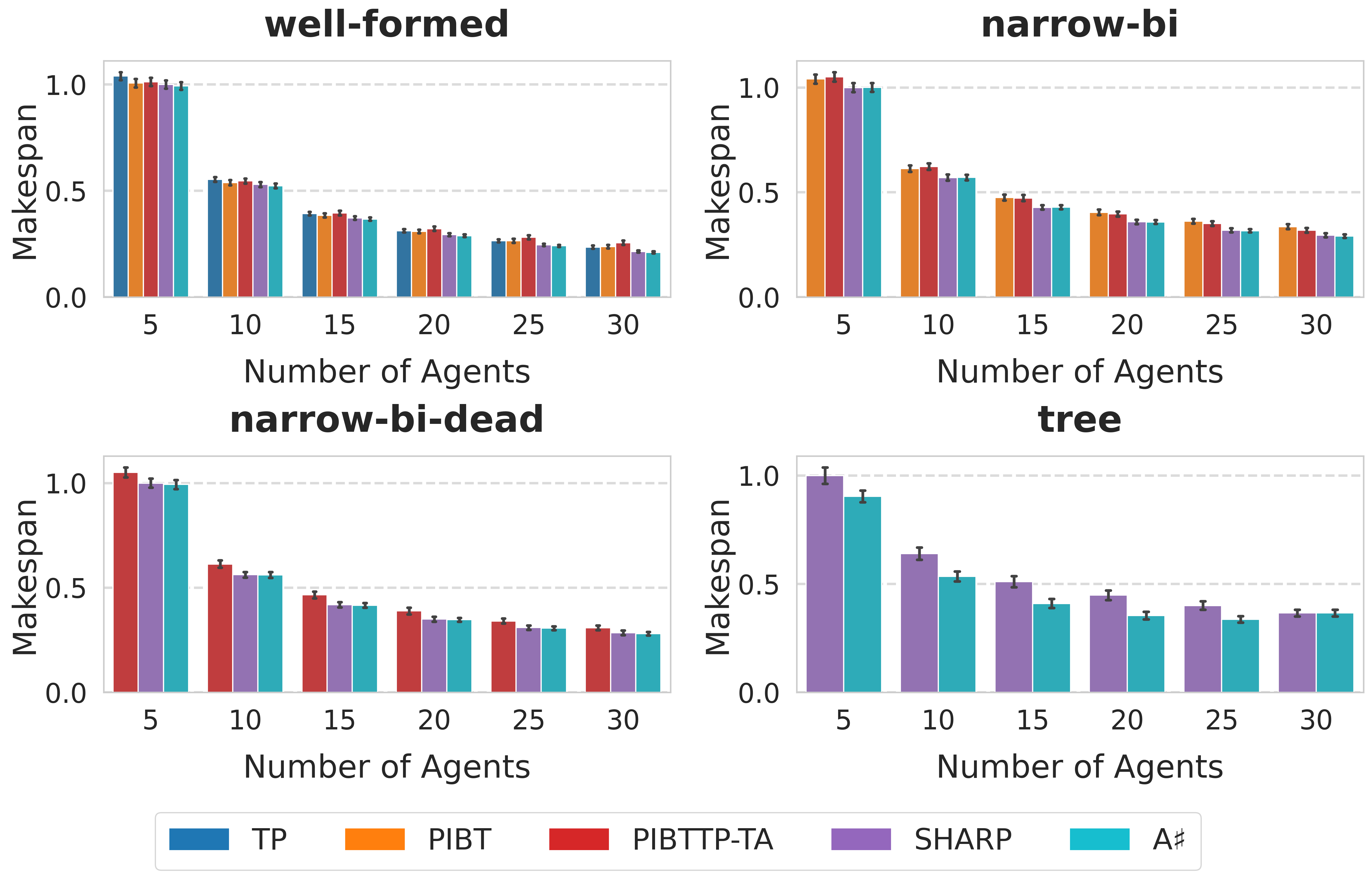}
\caption{Makespan vs.\ agents ($\lambda=3.0$).}
\end{subfigure}
\caption{Normalized makespan under the same per-map, per-axis anchor normalization as Figure~\ref{fig:service-performance}. Lower is better. Error bars denote one standard deviation over paired seeds; absolute anchors for constrained maps are not reported in the paper.}
\label{fig:makespan-performance}
\end{figure}

\subsection{Computation Time}

Table~\ref{tab:time} reports representative computation times.
In these representative settings, \asharpalg is not consistently more expensive than \sharpalg; the results are consistent with shorter retreat commitments offsetting some additional Haven-availability checks.
The only slowdown shown here occurs in the tree setting with $|A|=|H|=30$, where \asharpalg has no alternative Haven and the remaining difference is overhead.
Unlike Table~\ref{tab:effect-summary}, these timing ratios include all 144 \sharpalg/\asharpalg configurations, including those with no alternative Haven.
Across all 144 \sharpalg and \asharpalg configurations, the median \asharpalg/\sharpalg ratio is 0.83 for milliseconds per step, with a range from 0.40 to 1.03.
Because \asharpalg often shortens the number of simulated timesteps, the corresponding median ratio for total simulator computation time per completed run is 0.81, with a range from 0.33 to 1.03.
All algorithms were evaluated in the same Python simulator using paired task sequences, without GPU acceleration; reported times are simulator wall-clock measurements and should be interpreted as implementation-level timing rather than asymptotic evidence or deployed-product performance.
The host environment was Ubuntu~24.04.3 under WSL2 on an Intel Core i7-14650HX, with 24 logical processors exposed and 15~GiB RAM.
Simulation jobs ran in an Ubuntu~22.04 Docker container with Python~3.11; each simulator process was single-threaded, and independent runs were parallelized.

\paragraph{Research artifact.}
The source code, machine-readable maps, paired-task generator, experiment scripts, and reproduction instructions are publicly available at \url{https://github.com/abigworld1/A-SHARP}.
The source code is provided for evaluation and research purposes under the repository-specific terms.

\begin{table}[!htbp]
\caption{Representative computation time in ms/step. Dashes indicate failure before all tasks are completed.}
\label{tab:time}
\centering
\begin{tabular}{lrrrrr}
\toprule
Map / setting & TP & PIBT & PIBTTP-TA & \sharpalg & \asharpalg \\
\midrule
well-formed, 10 agents, $\lambda=1.0$ & 0.99 & 0.65 & 0.70 & 2.20 & 1.96 \\
well-formed, 30 agents, $\lambda=3.0$ & 6.33 & 2.04 & 2.36 & 26.2 & 19.6 \\
tree, 10 agents, $\lambda=1.0$ & -- & -- & -- & 15.1 & 5.98 \\
tree, 30 agents, $\lambda=3.0$ & -- & -- & -- & 137.5 & 141.4 \\
\bottomrule
\end{tabular}
\end{table}
\FloatBarrier

\section{Discussion and Limitations}

\paragraph{Confidential task-performance values.}
To avoid misinterpretation as performance specifications of an affiliated company's commercial products, we do not report our measured absolute makespan and service-time values on the constrained maps.
The constrained-map task-performance results should therefore be interpreted as paired relative algorithmic comparisons in simulation rather than benchmarks of a deployed product; the public well-formed benchmark and computation times are reported in absolute units.
The public artifact nevertheless permits independent reruns; values obtained in another environment are not product-performance claims.

\paragraph{Map and demand coverage.}
Each structural category is represented by one map, and the 100 seeds vary initial Haven assignments and task sequences rather than map topology.
Task endpoints are sampled uniformly from the stated candidate set.
The results therefore establish behavior on the four tested layouts and demand model, not across a distribution of tree-like or narrow-biconnected topologies or nonuniform warehouse demand patterns.

\paragraph{Execution model.}
The guarantee assumes deterministic discrete-time execution and a centralized reservation table.
Unexpected delays, mechanical failures, and localization errors are outside the current theorem.
Existing robust, dynamic-environment, and external-agent MAPD work addresses related disturbances and moving agents outside the planned team \cite{lodigiani2023robust,flammini2024deadlocks,bonalumi2025external}.
Handling such disturbances in \asharpalg would require reservation repair, temporal slack, or online replanning policies layered on top of the present ownership-transfer mechanism.

\paragraph{Haven density.}
The current framework assumes distinct Havens and therefore $|A|\leq |H|$.
This is natural when robots have dedicated parking or charging cells, but denser fleets would require shared parking, queueing, or buffer-cell rules.

\paragraph{Heuristic choice.}
Nearest-pickup task selection and nearest-available-Haven selection are simple heuristics.
They are not the safety contribution; congestion-aware selection can be substituted as long as candidate Havens satisfy availability and commitments preserve the same invariants.
The service-time losses on some narrow-biconnected configurations are consistent with this limitation: a locally nearest Haven can attract traffic to a constrained area even when the ownership-transfer protocol itself remains safe.
This pattern points to replacing the selector, not to weakening the availability and pending-release protocol.
The same protocol can also constrain learned or optimization-based task and Haven selectors, provided they preserve the same commitment semantics and return only available Havens.

\section{Conclusion}

\asharpalg addresses the main difficulty created by dynamic retreat targets: future path reservations and persistent Haven ownership must be updated together.
The availability check and pending-release rule provide this ownership-transfer protocol while keeping the safety structure of safe-haven retreat planning intact.
Under explicit Haven structure conditions and SIPP planning assumptions, we proved that \asharpalg preserves exclusivity and reservation invariants and delivers every task in any finite release sequence.
In the experiments, both \sharpalg and \asharpalg achieved 100\% success across all tested configurations, showing that dynamic Haven ownership transfer did not weaken the robustness of the fixed-Haven baseline.
For makespan, \asharpalg was significantly better in 107 of 138 Haven-surplus configurations and never significantly worse than \sharpalg after Holm correction across all 138 primary comparisons; the largest map-level median gain was 16.7\% on the tested tree map.
These results show that completion-oriented safe-haven retreat can be made dynamic in narrow or dead-end-heavy layouts, while leaving congestion-aware, learned, or optimization-based Haven selection as a natural next step.

\appendix
\section{Detailed Proofs}
\label{app:proofs}

\begin{proof}[Proof of Proposition~\ref{prop:update}]
Let $h_{\mathrm{old}}=\eta_t(a)$.
If $h_{\mathrm{new}}=h_{\mathrm{old}}$, no ownership transfer occurs, so injectivity and disjointness are unchanged and no new pending-release Haven is created.
Otherwise, availability ensures $h_{\mathrm{new}}$ is not in another agent's exclusive set, so adding it to $X_t(a)$ and setting $\eta_t(a)=h_{\mathrm{new}}$ preserves injectivity and disjointness.
If $a$ is still located at $h_{\mathrm{old}}$, pending release keeps that vertex in $X_t(a)$, so other agents continue to treat it as blocked.
If $a$ has already left $h_{\mathrm{old}}$, releasing it cannot allow another agent to plan through a vertex occupied by $a$.
The update creates at most one pending-release old Haven for $a$ because the selected agent is removed from $I$ after a successful commitment and therefore cannot receive a second ownership update in the same assignment loop.
Across timesteps, a pending old Haven exists only until $a$ departs; post-execution cleanup then removes it before $a$ becomes eligible again.
\end{proof}

\begin{proof}[Proof of Lemma~\ref{lem:invariants}]
During the assignment phase at time $t$, each successful commitment is applied to the phase-local state before the next candidate is evaluated.
The argument below first shows that these phase-local execution and reservation conditions are preserved after each commitment, and then applies the execution and cleanup step to obtain the beginning-of-timestep invariants at $t+1$.
Consider one successful commitment inside the assignment loop.
\textsc{PlanFullPath} validates a candidate full path that starts at $\mathrm{pos}_a(t)$, uses only wait actions and graph edges, ends at the selected Haven, and avoids other agents' reservations and exclusive Havens.
Agent $a$'s replaceable vertex reservations for times $t'>t$ are ignored only during validation; commitment atomically removes those reservations and inserts the validated path while keeping the time-$t$ vertex reservation fixed.
The induced directed moves from time $t$ onward are exactly the consecutive entries of the inserted path.
Thus no stale self-reservations remain, the position-reservation consistency for $a$ is re-established by the inserted path, and planner soundness preserves vertex and edge-swap reservation constraints.
The selected agent is eligible, so it has no pending-release old Haven before the commitment by the state classification above.
Proposition~\ref{prop:update} handles the corresponding ownership update.
The reservation-exclusion invariant is also preserved.
The inserted path avoids every other agent's exclusive set, so agent $a$ does not newly reserve any vertex in $X_t(b)$ for $b\neq a$.
If the selected Haven becomes part of $X_t(a)$, availability ensures that no other agent has a committed reservation for it at any time $t'\geq t$.
If the old Haven remains pending in $X_t(a)$, the invariant already excluded other agents' reservations for that Haven before the update; if it is released, it is no longer subject to the invariant for $a$.
Because commitments in the same while loop are sequential, each later validation sees the reservations and exclusive Havens inserted by earlier commitments; repeated commitments therefore preserve the invariants by induction.
The final execution step preserves the reservation invariants because agents follow committed moves that match their reserved transitions, and pending-release Havens remain blocked to other agents until the owning agent has departed.
The post-execution release step removes an old Haven only after its owner has left it, so it cannot expose a currently occupied vertex to other agents.
Idle agents wait at their Havens.
\end{proof}

\begin{proof}[Proof of Lemma~\ref{lem:quiescence}]
In quiescence, no finite movement suffix remains.
By the self-availability observation after Definition~\ref{def:available-haven}, agent $a$'s current Haven is available to itself.
If the selected Haven is $\eta_T(a)$, the final segment returns to the same protected vertex after leaving it; $a$'s own exclusive set and own finite wait entries do not block its replan.
If $h\neq\eta_T(a)$, availability ensures that $h$ is neither owned by another agent nor occupied by another agent's committed future reservation.
For the first segment, agent $a$ leaves its current Haven $\eta_T(a)$ through a neighbor in $F$ and then follows a path in $G[F]$ to $s_\tau$.
The pickup-to-delivery segment stays inside $G[F]$.
For the final segment, the agent follows a core path from $g_\tau$ to a neighbor of the selected Haven $h$ and then enters $h$.
Thus the first and final segments have length at most $\mathrm{diam}(G[F])+1$, and the middle segment has length at most $\mathrm{diam}(G[F])$.
The constructed paths use no Haven vertices except $a$'s start Haven and the final selected Haven $h$; these may be the same vertex when $h=\eta_T(a)$.
They therefore avoid other agents' exclusive Havens and finite future non-Haven reservations by quiescence.
This construction proves that feasible segments exist; SIPP need not return these exact paths, but quiescent-segment completeness within the stated horizon implies that the three segment validations can succeed.
\end{proof}

\begin{proof}[Proof of Theorem~\ref{thm:complete}]
Suppose some released task remains unfinished forever.
Any assigned task has a finite committed path that reaches its delivery before ending at a Haven, so an unfinished task that persists forever must eventually remain pending.
Because only finitely many tasks are released, there is a time $T_0$ after which no new task is released.
Each successful commitment removes one pending task from the pending set and tasks are assigned at most once, so only finitely many successful commitments can occur.
If no successful commitment occurs at or after $T_0$, set $T_1=T_0$; otherwise let $T_1\geq T_0$ be a time after the final successful commitment.
Temporary validation failures during active traffic do not affect this argument: after $T_1$, all remaining committed finite suffixes are simply executed to their current Havens, and no new task release creates additional work.
By the reservation invariant, every non-idle agent after $T_1$ follows a finite suffix to its current Haven, so the system reaches a quiescent configuration at some later time $T_2$.
At $T_2$ at least one released task is still pending by the supposition above.
Every idle agent has at least one available Haven, namely its current Haven by Definition~\ref{def:available-haven} and the reservation-exclusion invariant.
Lemma~\ref{lem:quiescence} therefore implies that every idle agent can validate a full path for a pending task to at least one available Haven.
Because Lemma~\ref{lem:quiescence} quantifies over any pending task and any available Haven, it applies in particular to the nearest pending task and nearest available Haven selected by Algorithm~\ref{alg:asharp-workshop}; the greedy task-Haven choice therefore cannot block progress in quiescence.
Since the assignment loop is invoked while pending tasks remain, at least one candidate must validate and the loop must commit a task at or after $T_2$, contradicting the choice of $T_1$ as after the final successful commitment.
\end{proof}

\section{Restricted Fixed-Haven Baseline}
\label{app:fixed-pseudocode}

Algorithm~\ref{alg:sharp-workshop} records the fixed-Haven comparison used in the experiments.
It uses the same task selector, SIPP validation, and commitment order as Algorithm~\ref{alg:asharp-workshop}, but always retains the initial injective Haven assignment $\eta^0$.

\begin{algorithm}[!htbp]
\caption{Restricted fixed-Haven \sharpalg assignment loop}
\label{alg:sharp-workshop}
\footnotesize
\begin{algorithmic}[1]
\Require Decision time $t$; pending tasks $Q$; agents $A$; fixed injective Haven assignment $\eta^0$; static exclusive sets $X^0(a)=\{\eta^0(a)\}$; reservation table $\mathcal{R}$; segment horizon $T_{\max}$
\Ensure Updated $Q$, $\mathcal{R}$, and agent commitments
\State Add tasks released at time $t$ to $Q$; let $I$ be the idle or retreating agents
\While{$Q\neq\emptyset$ and $I\neq\emptyset$}
  \State $c^*\gets\bot$; $d^*\gets\infty$
  \For{each $a\in I$}
    \State $\tau\gets\arg\min_{\tau'\in Q}\mathrm{dist}(\mathrm{pos}_a(t),s_{\tau'})$
    \State $(\mathit{success},\pi)\gets\textsc{PlanFullPath}(a,\tau,\eta^0(a),t,\mathcal{R},X^0,T_{\max})$
    \If{$\mathit{success}$ and $\mathrm{dist}(\mathrm{pos}_a(t),s_\tau)<d^*$}
      \State $c^*\gets(a,\tau,\pi)$; $d^*\gets\mathrm{dist}(\mathrm{pos}_a(t),s_\tau)$
    \EndIf
  \EndFor
  \If{$c^*=\bot$} \State \textbf{break} \EndIf
  \State Let $(a,\tau,\pi)\gets c^*$ and keep $a$'s time-$t$ vertex reservation fixed
  \State Delete $a$'s reservations for $t'>t$ and insert the future entries of $\pi$ in $\mathcal{R}$
  \State Assign $\tau$ to $a$ and mark $a$ task-executing
  \State Remove the selected agent and task from $I$ and $Q$
\EndWhile
\State Execute one timestep along committed reservations
\end{algorithmic}
\end{algorithm}

\section*{Declaration on Generative AI}
During the preparation of this work, the authors used OpenAI Codex in order to: Paraphrase and reword, Improve writing style, and Grammar and spelling check.
The authors supplied and verified the scientific content and arguments, reviewed and edited all tool-assisted changes, and take full responsibility for the publication's content.

\bibliography{caipi_prl26}

\end{document}